\documentclass[journal]{IEEEtran}

\usepackage[T1]{fontenc}
\usepackage{booktabs}
\usepackage{enumitem}
\usepackage[square, comma, compress, numbers]{natbib}
\usepackage{amsfonts, amsmath, amssymb, amsthm, dsfont}
\usepackage[cmintegrals]{newtxmath}
\usepackage{array}

\usepackage{amsthm}
\newtheorem{theorem}{\textbf{Theorem}}
\newtheorem{corollary}{\textbf{Corollary}}
\newtheorem{definition}{\textbf{Definition}}
\newtheorem{lemma}{\textbf{Lemma}}

\usepackage{mathtools}

\DeclareMathOperator*{\argmin}{arg\,min}
\usepackage[normalem]{ulem}
\usepackage{xcolor}
\usepackage{hhline}

\usepackage{enumitem}

\usepackage[pscoord]{eso-pic}% The zero point of the coordinate system is the lower left corner of the page (the default).
\newcommand{\placetextbox}[3]{% \placetextbox{<horizontal pos>}{<vertical pos>}{<stuff>}
  \setbox0=\hbox{#3}% Put <stuff> in a box
  \AddToShipoutPictureFG*{% Add <stuff> to current page foreground
    \put(\LenToUnit{#1\paperwidth},\LenToUnit{#2\paperheight}){\vtop{{\null}\makebox[0pt][c]{#3}}}%
  }%
}%

\begin{document}

\placetextbox{0.5}{0.97}{\texttt{This work has been submitted to the IEEE for possible publication. Copyright may be }}%
\placetextbox{0.5}{0.954}{\texttt{transferred without notice, after which this version may no longer be accessible.}}%

\title{Relay Incentive Mechanisms Using Wireless Power Transfer in Non-Cooperative Networks}

\author{Winston~Hurst,~\IEEEmembership{Student Member,~IEEE,}
        and~Yasamin~Mostofi,~\IEEEmembership{Fellow,~IEEE}% <-this % stops a space
\thanks{Winston Hurst and Yasamin Mostofi are with the Department of Electrical and Computer Engineering, University of California, Santa Barbara, USA (email: \{winstonhurst, ymostofi\}@ece.ucsb.edu). This work was supported in part by ONR award N00014-23-1-2715.}% <-this % stops a space
}%\thanks{Manuscript received April 19, 2005; revised August 26, 2015.}}

\maketitle

\begin{abstract}
The advances of 6G systems have prompted the study of new communication paradigms, including relay networks enabled by wireless power transfer (WPT). While existing literature focuses on the cooperative case, this paper considers the scenario where a source must offer payment in the form of WPT to induce one of several non-cooperative, battery-powered user equipments (UEs) to act as a relay. In this non-cooperative setting, the utility-maximizing UEs will try to extract as much energy from the source as possible. We propose a protocol based on a reverse auction that enables the source to determine which candidate UE to select as the relay and the amount of energy to be transferred as payment, even when the channel quality between the candidates and the destination is unknown to the source. We first examine the performance of the system under a second-price, sealed bid auction, i.e., a Vickrey auction. We prove that our protocol achieves the best possible outage probability and point out ways to mitigate the gap in energy efficiency when compared to a cooperative baseline. We then examine system performance when a Myerson auction is used, which is known to maximize an auctioneer's utility. To ensure tractability of the Myerson pricing mechanism, a regularity condition must be satisfied, and we extend the established forward auction regularity result to the reverse auction setting by providing a mathematical characterization of the regularity condition and the related pricing mechanism for reverse Myerson auctions. We then provide detailed proofs of the regularity of the WPT-based auction for both lognormal fading and Rayleigh fading. To validate and illustrate the analytical findings, we include the results of extensive numerical experiments, which elucidate the impact of system parameters (e.g., number of candidates, environment geometry) on energy efficiency and outage probability. These studies highlight conditions under which one type of auction is preferred over the other and give guidance on system design. Overall, the auction-based system offers a foundation for improved energy efficiency and communication reliability in non-cooperative environments.
\end{abstract}

% \begin{IEEEkeywords}
% \end{IEEEkeywords}
%%%%%%%%%%%%%%%%%%%%%%%%%%%%%%%%%%%%%%%%%%%%%%%%%%%%%%%%%%%%%%%%
%%%%%%%%%%%%%%%%%%%%%%%%%%%%%%%%%%%%%%%%%%%%%%%%%%%%%%%%%%%%%%%%

\IEEEpeerreviewmaketitle

\vspace{-0.15in}
\section{Introduction}\label{sec:introduction}

\IEEEPARstart{A}{s} 6G standards come into clearer focus, a number of key technologies, paradigms, and challenges have received increasing attention \cite{IMT-2030}. Higher communication frequencies, including mmWave and THz communications, will enable massive data rates, but they will also require new network topographies to account for propagation phenomena, such as high penetration loss and sparse multipath effects, not commonly present at lower frequencies \cite{Hemadeh2018CST}. As the number of connected devices continues to grow rapidly, there is increased demand for distributed architectures and frameworks, such as federated learning and edge computing, that avoid the overhead of fully centralized systems \cite{Nguyen2022IoTJ, Wang2023CST}. Wireless power transfer (WPT) has also received attention as a potential enabler for networks of unpowered sensors, among other applications \cite{Ponnimbaduge2018CST, Ozyurt2022Access}. Finally, running as a common thread through all these trends is the need to keep systems energy efficient, particularly in the face of climate change \cite{Kamran2024COMSNETS}.

In this paper, we examine a scenario at the confluence of these factors. Specifically, we consider a source node which needs to offload data to a remote access point (AP) with which the source does not have a direct line-of-sight (LOS) channel, due to some blockage. For high-frequency channels, this blockage creates significant penetration loss, making direct communication with the AP costly in terms of energy consumption or potentially infeasible due to limits on transmit power. The source may try to enlist another nearby user equipment (UE) to act as an intermediate relay. However, to convince the non-cooperative entity to provide the relaying service, the source can use wireless power transfer (WPT) as a form of payment. Fig.~\ref{fig:intro_setup} illustrates the problem scenario.

Our objective is to design a protocol which ensures that the source can communicate with the AP in an energy-efficient manner while accounting for the individual rationality of the participating entities. We approach the problem from a game-theoretic perspective and propose an auction-based protocol which achieves the desired system behavior. We next present a brief literature review before stating our contributions.

\begin{figure}
    \centering
    \includegraphics[width=0.95\linewidth]{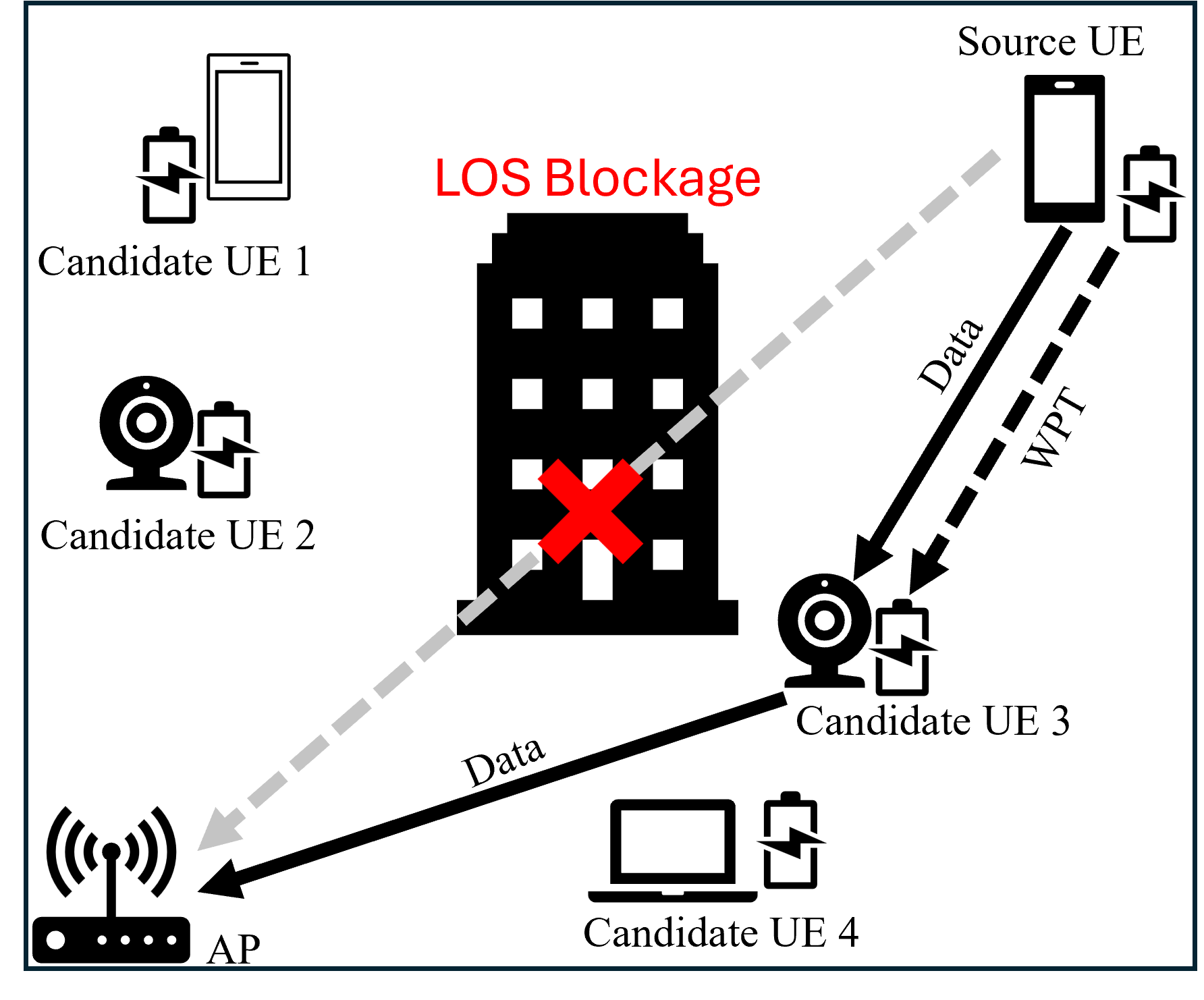}
    \vspace{-0.1in}
    \caption{Overview of problem scenario. A source UE must offload data to the AP, but the LOS path is blocked, resulting in poor link quality. The source must induce a candidate UE to cooperate with a payment in the form of WPT.}
    \label{fig:intro_setup}
    \vspace{-0.2in}
\end{figure}

\vspace{-0.15in}
\subsection{Literature Review}
%WPT - convince the reader this is a legitimate technology
WPT involves the transfer of energy from one device to another without a wired connection, and while several technologies fall under the umbrella of WPT, we specifically consider RF-based WPT. In recent years, it has found many applications, from smart home ecosystems to industrial manufacturing \cite{mulders2022Sensors}. For example, RF WPT can be used in Internet of Things (IoT) networks to charge batteries on small, remote devices, possibly with the aid of unmanned vehicles \cite{Xie2021TGCN}. Experimental results demonstrate received powers in the microwatt or milliwatt ranges~\cite{CANSIZ2019292}, and although comparatively inefficient compared to, e.g., inductive coupling, RF WPT can operate across greater distances (up to several km) \cite{kim2014PIEEE}. Furthermore, as communication networks grow denser, harvesting ambient RF energy becomes a more attractive possibility, which has prompted research into hardware improvements \cite{fallahpour2022ICEE,Bougas2024MOCAST}.

%SWIPT
In the context of communication systems, the past several years have seen increased interest in simultaneous wireless information and power transfer (SWIPT), particularly in the context of IoT networks and cooperative relaying (CR) \cite{IoTJ2021LiWangLiuLiPengPiranLi, TWC2016LiuKimKwakPoor, Zhai2020TVT,Hossain2019Access, Ropokis2022TVT, Ashraf2021Access}. 
%start summarizing the literature
\cite{IoTJ2021LiWangLiuLiPengPiranLi} considers a single BS using non-orthogonal multiple access (NOMA) to communicate with a near user and a far user while aided by a set of energy-harvesting relays. It then focuses on calculating the outage probability and finding the optimal power allocation while accounting for residual hardware errors and channel estimation errors. A SWIPT setup using antenna diversity receives attention in \cite{TWC2016LiuKimKwakPoor}, which compares the performance of full-duplex to half-duplex relaying. The energy-harvesting relay node employs two batteries for greater efficiency, one for energy harvesting, the other to power the transmission. The outage probability and optimal power splitting ratios are found. In \cite{Zhai2020TVT}, the authors propose a wireless communication system integrating WPT and CR to enhance node operations and communication reliability. They explore both orthogonal and non-orthogonal CR schemes where multiple sources communicate with destinations via an energy-harvesting relay, adapting transmission strategies based on the relay decoding status. Numerical analysis demonstrates that the non-orthogonal CR scheme outperforms the orthogonal scheme in certain scenarios, influenced by parameters such as transmit powers, data rates, and terminal distances. \cite{Ropokis2022TVT} proposes a cooperative communication scheme based on WPT-enabled relaying that uses both time splitting (TS) and antenna switching (AS) paradigms, as well as self-energy recycling (SER). Splitting and switching ratios are optimized to maximize the communication rate, and the scheme is shown to outperform comparable designs. A comprehensive survey of WPT for CR is found in \cite{Ashraf2021Access}. 

The aforementioned body of work focuses on cooperative settings, so that the source and the relay are generally analyzed or optimized as a single unit. In contrast, we consider a non-cooperative scenario, where other nodes must be incentivized to provide the relaying service. This approach allows for WPT-enabled relaying to occur dynamically, rather than relying on the deployment of dedicated relays or assumptions of cooperation. As opposed to existing literature, this perspective is well-suited for analysis via the tools of game theory, and we will make particular use of classical results in auction theory.

Many aspects of communication systems are well modeled as the interaction of self-interested agents. Consequently, game theory and auctions have long been used to approach a wide range of problems, including resource allocation, relay selection, and power management \cite{GTforWCN2012HanNiyatoSaadBasarHJorungnes, TRO19_MuralidharanMostofi}. Many works sit specifically at the intersection of WPT and game theory \cite{TWC2018ZhangDuChengLongLeung, ISIT2014Chen, Ni2019ICCT, Cheng2023TNSM, TWC2014Ding}, including several that have used non-cooperative games to understand how to manage interference, which, while detrimental to communication, may be harvested for additional energy \cite{TWC2018ZhangDuChengLongLeung, ISIT2014Chen}. Recent years have also seen the use of auctions to determine prices for WPT services. For example, \cite{Ni2019ICCT} purposes a peer-to-peer double auction mechanism to extend the battery life of mobile devices. More recently, \cite{Cheng2023TNSM} considers the interaction between a network operator and a power supplier which uses WPT to charge network nodes. A Stackleberg game models their interaction, and the game framework enables price design for the WPT service while accounting for both the costs and battery efficiency in the network. Auctions have also appeared in the context of WPT-enabled relaying. For example, in \cite{TWC2014Ding}, a single relay serves multiple source-destination pairs and must decide how to allocate the power received from the sources for communication to the destinations. An auction-based approach is formulated to determine how to allocate the received power among the several destinations, and its Nash equilibria are studied. In all these works, however, the auction payments are monetary.

While WPT and relay have appeared frequently in conjunction with auctions and other game-theoretical approaches, our work is unique in two fundamental ways. First, we use the auction to bring about WPT-enabled relaying in a non-cooperative setting. In this respect, our work is most similar to \cite{TWC2014Ding}. However in that paper, the auction determines how to allocate power which has already been harvested by a relay, while in our paper, the auction is used to select a suitable relay and determine the amount of energy to transfer via WPT. Second, we use energy itself as the medium of exchange. 

%why this is challenging
In non-cooperative networks, designing an effective relaying protocol is challenging due to the self-interested nature of the UEs, which aim to maximize their utility by extracting as much energy from the source as possible. Additionally, the source lacks knowledge of the channel quality between UEs and the AP, complicating the selection of an efficient relay. Furthermore, the protocol must balance energy efficiency and communication reliability, while avoiding excessive computation or communication overhead, making the design of incentive-compatible protocols in non-cooperative settings a complex yet critical task. Our work addresses each of these challenges, and we next state our contributions.

\vspace{-0.1in}
\subsection{Contributions}
This paper provides several contributions to the literature, as enumerated below:
\begin{itemize}%[leftmargin = 0.15in]
    \item We present the problem of using WPT to induce energy-efficient relaying in a non-cooperative setting by modeling the battery-powered relaying candidates as rational, utility-maximizing agents that must be motivated to act as relays through a payment of energy. The source must choose which candidate to use and how much energy to send. However, it lacks knowledge of the channel quality between the candidates and the AP, and the non-cooperative candidates may not reveal this information as they seek to maximize their own utilities. We propose the use of WPT-based auctions, implemented in a simple protocol, to achieve the desired cooperation, and for comparison, present a cooperative baseline.

    \item For Vickrey WPT auctions (VWAs), we mathematically prove that the outage probability is minimized, that is, communication will occur whenever possible, as is in the case of the cooperative baseline. We show that this outage probability decays exponentially as the number of candidates grows, and we mathematically characterize the energy efficiency gap compared to the cooperative baseline. 

    \item As an alternative, we propose Myerson WPT auctions (MWAs), which significantly improve energy efficiency. For the pricing mechanism to be computationally tractable, these auctions must satisfy a regularity condition. We first rigorously develop a novel extension of the classical regularity result to reverse Myerson auctions, then mathematically proving the regularity of the MWA with both lognormal and Rayleigh fading. We then mathematically characterize the outage probability, and show that, although the outage probability increases compared to the VWA, the gap in outage probability decays exponentially as the number of candidates grows.
    
    \item Finally, we present several numerical studies that illustrate how system parameters (e.g., the number of candidates, the fading parameters, and the auction type) affect the outage probability, the source's energy consumption, and the energy harvested by the candidate. Comparing the results for the VWA and MWA, we offer insights into the tradeoffs of each auction type and give guidance when one auction type is preferred over the other.
\end{itemize}

Aspects of this work appear in our conference paper \cite{hurst2024GLOBECOM}, which focuses on Myerson auctions with lognormal fading. As well as providing detailed proofs omitted from \cite{hurst2024GLOBECOM} due to space constraints, this paper analyzes Vickrey auctions and further examines the Myerson auction under a Rayleigh fading model. Furthermore, we provide an in-depth comparison of the Vickrey and Myerson auctions and give guidance on when one type of auction is preferred over the other.

The rest of the paper is organized as follows. In Section~\ref{sec:modeling}, we lay out the communication and WPT models used throughout the paper, and in Section~\ref{sec:auction}, we present our auction-based protocol. We then analyze systems performance of the protocol when using a Vickrey auction in Section~\ref{sec:vickrey} and when using a Myerson auction in Section~\ref{sec:myerson}. Section~\ref{sec:results} then corroborates the analytical results and illustrates the impact of various system parameters, and Section~\ref{sec:conclusion} concludes the paper.

\section{System Modeling}\label{sec:modeling}

%%%%General Setup%%%
% mobile UAV
% other nodes, assumed fixed over timeline of interest
%
We consider a system consisting of a source UE; a set of $n$ battery powered UEs, referred to as candidates; and a single access point (AP). The source is located at $q_s\in\mathbb{R}^2$, and the candidates are located at positions $q_i\in \mathbb{R}^2,\,\forall i\in \mathcal{N}$, where $\mathcal{N}=\{1,...,n\}$. We treat $q_i$ as independently distributed over the region of interest, and without loss of generality, we assume the access point is located at the origin. The system is illustrated in Fig.~\ref{fig:math_fig}.

\begin{figure}
    \centering
    \includegraphics[width=0.9\linewidth]{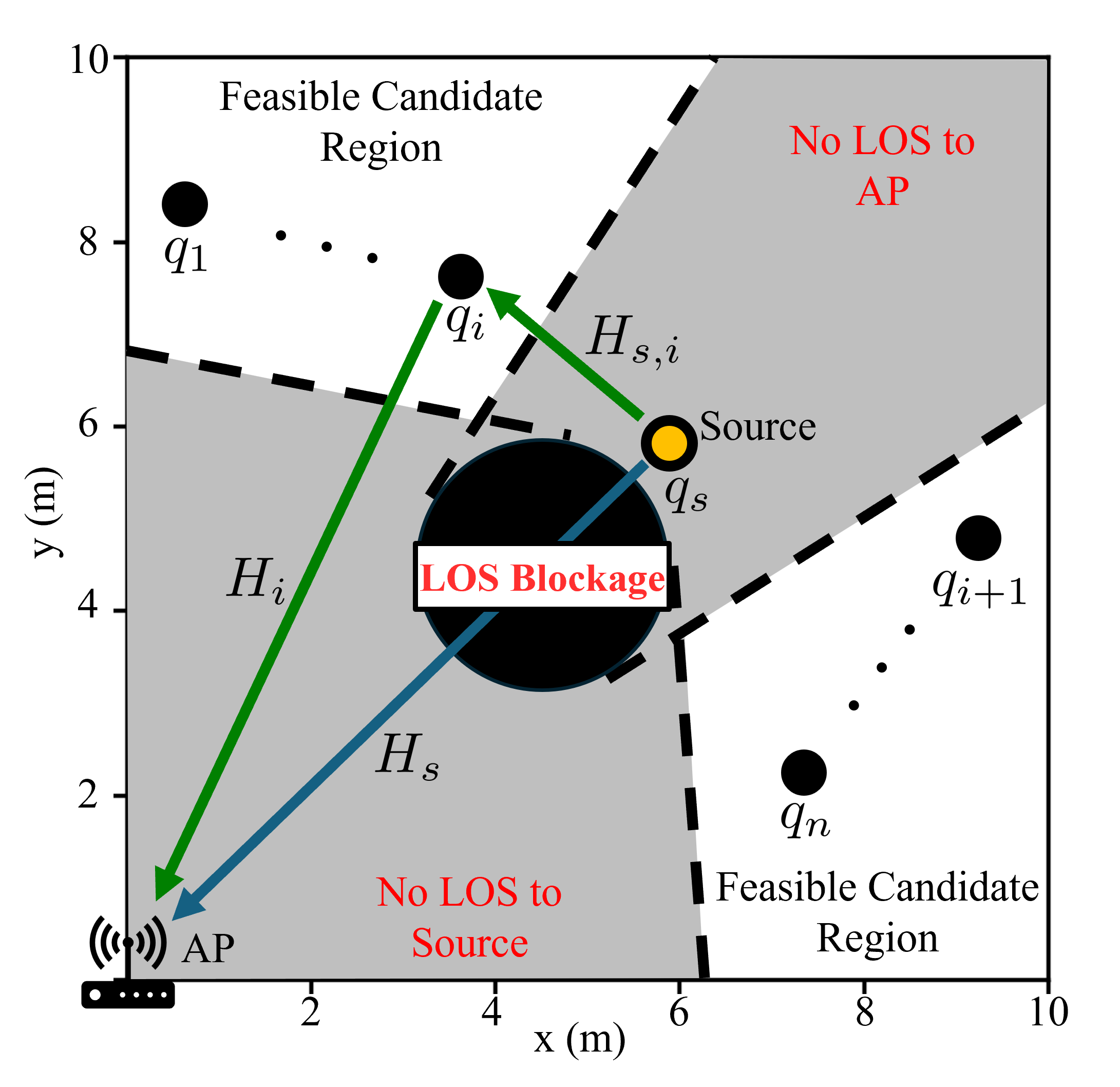}
    \vspace{-0.1in}
    \caption{System setup for the WPT-enabled relaying scenario. The orange dot gives the location of the source, $q_s$, the black dots give the locations of the candidates, $q_i$, and the AP is at the origin. The large black circle in the center indicates the blockage, and the gray region indicates the area where the channel to either the source or the AP is NLOS. The arrows represent channels, with the green arrows representing channels used if UE $i$ acts as the relay. The channels are labeled with corresponding channel power variables.}
    \label{fig:math_fig}
    \vspace{-0.2in}
\end{figure}

The source must transfer data to the AP in a timely manner. To do so, it may either try to send the data itself, or it may try to induce one of the candidates to serve as a relay, incentivizing cooperation using WPT. In the rest of this section, we provide our communication model, including channel modeling, as well as our model for WPT.

\begin{table*}
    \centering
    \begin{tabular}{|c|c||c|c||c|c|}
    \hline
         $q_s$ & Source position & $q_i$ & $i^\text{th}$ candidate position & $\zeta$ & Min. required received power\\
         \hline
          $\sigma^2$ & Receiver noise power& $H_s$ & Channel power, source to AP & $H_i$ & Channel power, candidate $i$ to AP \\
          \hline
           $H_{s,i}$ & Channel power, source to candidate $i$ & $P_{\text{tot}}$& Total source transmit power & $P_{\text{max}}$ & Max. transmit power\\
           \hline
           $D$ & Bandwidth-normalized data requirement & $T$ & Required transfer time & $\tilde{\alpha}_i$ & Overall WPT efficiency for candidate $i$\\
           \hline
           $P_s$ & Min. TX power, source to AP & $P_{s,i}$ & Min. TX power, source to candidate $i$ & $P_{i}$& Min. TX power, candidate $i$ to AP\\
           \hline
    \end{tabular}
    \caption{Important system variables}
    \label{tab:impt_vars}
    \vspace{-0.2in}
\end{table*}

%%Comms Modeling
\vspace{-0.1in}
\subsection{Channel Modeling and Communication}
Let the source-to-AP, candidate-to-AP, and source-to-candidate wireless communication channel power (square of the modulus of the complex-valued channel gain) be denoted $H_s$, $H_i$, and $H_{s,i}$, respectively. We assume that these are sufficiently well-estimated by both the transmitter and receiver, so that the source knows $H_s$ and $H_{s,i},\,\forall i \in\mathcal{N}$, and the $i^{\text{th}}$ candidate knows $H_{s,i}$ and $H_i$. However, the source and candidates do not know channel coefficients for channels which they do not use, so that the source does not know $H_i,\,\forall i \in \mathcal{N}$, and the $i^{\text{th}}$ candidate does not know $H_s$ and $H_j,\;\forall j \in\mathcal{N}\setminus i$. Thus, in auction terminology, the power of the candidate-to-AP channel, $H_i$, constitutes \textit{private information}.

As in \cite{hurst2024GLOBECOM}, we model each channel as composed of path loss and small-scale fading components, with large-scale shadowing modeled by using different parameters for LOS and non-LOS (NLOS) channels. More specifically, for $l\in\{s, i, si\}$ we have $ H_l = H_{l, \text{PL}} H_{l, \text{f}}$, where $H_{l, \text{PL}}$ is the path loss component, and $H_{l,\text{f}}$ is the fading component, commonly modeled as a lognormal, Rayleigh, or Rician random variable \cite{Samimi2016VTC}. The path loss component is modeled as $H_{l, \text{PL}} = K d_l^{-\eta}$, with $K$ and $\eta$ the path loss intercept and exponent, respectively, and $d_l$ the distance between the transmitter and receiver. Path loss parameters $K$ and $\eta$ take on different values for the LOS and NLOS cases, and we distinguish these values with LOS and NLOS superscripts (e.g., $K^{\text{LOS}}$, $K^{\text{NLOS}}$). We will assume that $H_s$ is a NLOS link, motivating the need to find a relay, and that $H_{i}$, $H_{s,i}$ are LOS links, as otherwise, candidate $i$ would almost surely provide no advantage over the direct NLOS link.

To model channel capacity, we rely on the Shannon-Hartley theorem, so that the maximum achievable spectral efficiency over the channel is
% $ r_{\text{max}} = \log_2\left(1+P_{TX}H_l/\sigma^2\right)$, 
\begin{equation}\label{eq:max_rate}
    r_{\text{max}} = \log_2\left(1+P_{TX}H_l/\sigma^2\right),
\end{equation}
where $P_{\text{TX}}$, $H_l$, and $\sigma^2$ are the transmit power, channel power, and receiver noise power, respectively. We further assume that the source and all UEs have a maximum transmit power $P_{\text{max}}$. Consequently, for transmission with spectral efficiency $r$ to be successful, we must have $r \leq \log_2\left(1+P_{\text{max}}H_l/\sigma^2\right)$.

\vspace{-0.15in}
\subsection{Data Transfer Requirements}
\vspace{-0.05in}
The source must transfer a given amount of data to the destination within a fixed amount of time, $T$. This models scenarios with hard time-related constraints, such as vehicular systems (V2X), AR/VR, or industrial automation. We do not assume that satisfying the time constraint is always feasible, and thus we examine system performance in terms of both energy consumption and outage probability. We note that using the maximum allowed time, $T$, results in the lowest energy usage, so that, without loss of generality, we will assume the full time $T$ is always used for the relay procedure.

Let $D$ represent the bandwidth-normalized amount of data (in bit/Hz) that must be transferred. The specification of both $T$ and $D$ induces a spectral efficiency requirement of $r > D/T\,$(bit/s)/Hz, so that for the source to transfer the data directly to the AP, its minimal transmit power is $P_s =  \zeta/H_{s}$, where $\zeta = (2^{D/T}-1)\sigma^2$ is the minimum required received power to communicate at a rate of $D/T$. For deep fades (poor channel conditions) or short transmit times, this value may be quite large, and may in fact exceed $P_{\text{max}}$. Alternatively, the source may use WPT to incentivize a candidate to act as a decode-and-forward relay to the AP. 

\vspace{-0.15in}
\subsection{Relaying and WPT}
Relaying data to the AP via one of the candidates requires two transmissions, and we assume that, due to either the use of orthogonal frequency bands or full-duplex relaying capabilities, additional delay due to the relaying operation is negligible. First, to transmit the data to the candidate, the source must transmit with a power of at least $P_{s,i} = \zeta/H_{s,i}$.
Then, the candidate must transmit the data to the AP with a power of at least $P_i =  \zeta/H_{i}$. To induce cooperation, however, the source must also send enough energy to the candidate to at least cover the cost of transmission. Let $P_{\text{WPT}}$ be the power allocated for wireless power transfer. The power harvested by the $i^\text{th}$ UE, $P_{\text{harv},i}$,  depends on a number of factors, including the channel between the source and the candidate, $H_{s,i}$, the effective aperture of the receiver antenna, $A_r$, and the energy-harvesting circuitry efficiency, $\alpha$ \cite{Hossain2019Access}. Formally, we have $P_{\text{harv},i} = P_{\text{WPT}}H_{s,i}A_r\alpha = \Tilde{\alpha}_iP_{\text{WPT}}$
where $\Tilde{\alpha}_i = H_{s,i}A_r\alpha$ is the overall efficiency of the WPT.

We consider power-splitting WPT \cite{Zhou2012GLOBECOM} with simultaneous wireless information and power transfer (SWIPT), so that if the source chooses candidate $i$ as the relay, its total power consumption is $P_{\text{tot}} = P_{s,i} + P_{\text{WPT}}$, with $P_{s,i}/P_{\text{tot}}$ the receiver's power splitting ratio. \textbf{The key question addressed in this paper is this: how can the source determine which candidate to use (if any) and total transmit power, $P_{\text{tot}}$, in this non-cooperative setting?}

\section{An Auction Formulation}\label{sec:auction}
We view the source and all candidates as rational decision makers seeking to non-cooperatively maximize their individual utilities. The source and the candidates both seek to minimize their net energy expenditure. Additionally, the source must ensure that the data is offloaded within the time limit, $T$. These preferences are captured in the following utility functions for the source and candidates, denoted as $U_s$ and $U_i$, respectively:
\begin{equation}\label{eq:utilities}
\begin{split}
    U_{s} &= Cz -T P_{\text{tot}},\\
    U_i &= \begin{cases}
        0 & \text{If } \rho \neq i\\
         T (\tilde{\alpha}_i (P_{\text{tot}} - P_{s,i}) - P_i) & \text{If } \rho = i
    \end{cases},
\end{split}
\end{equation}
where $C$ is a large reward modeling successful data transmission to the AP within the time budget, $z\in\{0,1\}$ indicates whether this is actually achieved, and $\rho$ is the index of the candidate chosen for relaying (if any). To ensure the source prefers communicating whenever possible, we have $C>T P_\text{max}$. Smaller values of $C$ emphasize energy efficiency, while larger values of $C$ place greater weight on ensuring successful communication. We also assume that the transmission is sufficiently spatially constrained so that if a candidate is not chosen as the relay, it cannot harvest energy, though relaxing this assumption is an important direction for future work.

A reverse auction, sometimes called a procurement auction, naturally frames the considered scenario. In this framework, the source looks to ``buy" relaying services from the candidates, who bid for the opportunity to provide the service. In practice, the communication protocol proceeds as follows: (1) The source broadcasts a request for bids; (2) The candidates respond with bids; (3) the source broadcasts its choice of the winner, $\rho$, as well as the price, $P_\text{tot}$; (4) the winner acknowledges; (5) the relaying process begins. This process disallows extended bargaining between the source and candidates, which is desirable given the potential costs in both time and energy that could result. For ease of exposition, we consider bids as proposed values of the total source transmit power, $P_\text{tot}$. This is equivalent to bids in terms of total source energy usage, $TP_\text{tot}$, since the source will always use the full time to transmit. Furthermore, since both the source and candidate know the value of $H_{s,i}$, this is equivalent to expressing bids as proposed values of $P_\text{WPT}$. We next present details of the auction-based protocol.

\vspace{-0.1in}
\subsection{Auction Components}
We next review the components of an auction and their corresponding concepts in our WPT-enabled relaying setting.

%valuations
\textit{Valuations:} In a reverse auction, the valuations, $\mathcal{V} = \{v_0, ..., v_n\}$ indicate how much the service is worth to the participants, where $v_0$ is the value of the service to the procurer, and $v_i,\,\forall i \in \mathcal{N}$ is the cost that would be incurred by the bidders if they had to provide the service.

Thus, the valuation of the source (the service procurer) is the power required to directly transmit the data to the AP. Recall that $P_s$ is the minimum power required for direct transmission to the AP. If $P_s \leq P_{\text{max}}$, then $v_0 = P_s$. However, if $P_s > P_{\text{max}}$  the source cannot transmit the data directly within the time frame. In this case, we set $v_0 = P_{\text{max}}$, with the interpretation that this is the maximum power the source could possibly use for WPT payment. Consequently, $v_0 \leq P_{\text{max}}$ always holds.
%\footnote{If alternatively, the source's valuation is ignored when $P_s > P_\text{max}$, a number of issues arise.}

The valuations of the candidates may be similarly modeled. If the source elects to send data via the $i^{\text{th}}$ candidate, it must transmit with a power of at least $P_{s,i} + P_i/\tilde{\alpha}_i$, which ensures non-negative utility for the $i^{\text{th}}$ candidate. Thus, $v_i = P_{s,i} + P_i/\tilde{\alpha}_i$. Note that, to facilitate comparison, the valuation is expressed in terms of the power required by the source, not the power used by the candidate. Importantly, there is a one-to-one mapping between $v_i$ and the private information $H_i$, so that bidding $v_i$ is equivalent to revealing $H_i$. 

%bids
\textit{Bids:} The candidates' valuations depend on private channel information, $H_i$, which is not known by the source. Thus, in a general setting, the candidates could provide bids that differ from their true valuations, as their goal is to maximize their profit. However, this paper considers two specific auction rules (\textit{mechanisms}) under which bidding the true valuation maximizes a candidate's utility. These \textit{incentive-compatible} mechanisms have been well-studied in game theory literature and applied across a variety of domains~\cite{AGT2007}. As we discuss in greater detail later in this section, using a mechanism that results in the candidates' biding their true valuations differs in subtle but important ways form assuming that the source knows all the valuations beforehand or that the candidates act cooperatively with the source.

%mechanism
\textit{Mechanisms:} A mechanism is defined by two functions, $\rho(\mathcal{V})$ and $X(\mathcal{V})$, which take as arguments the set of valuations, $\mathcal{V}$, and determine the winner of the auction and the payment to the winner, respectively. In the context of the WPT scenario, $\rho(\mathcal{V})$ indicates which of the candidates will act as the relay, and $X(\mathcal{V})$ gives $P_{\text{tot}}$. Note that $\rho(\mathcal{V}) = 0$ indicates the source will not try to relay through any of the candidates. If that is the case and $P_s \leq P_{\text{max}}$, the source successfully sends the data directly. Otherwise, a communication outage occurs, and we assume that the source, knowing that the transmission will be unsuccessful, uses $P_{\text{tot}} = 0$. More generally, it will not transmit at all if communication cannot occur within the given time frame through any of the candidates or itself.

\vspace{-0.1in}
\subsection{System Metrics}
The trade-off between communication reliability and energy efficiency constitutes a fundamental challenge in our communication system, and we consider the outage probability and energy consumption as the core performance metrics of our proposed protocol. We next discuss each of these in more detail.

\subsubsection{Outage Probability}
For a given mechanism, the outage probability is the probability that the mechanism does not successfully transfer the data within the time limit, $T$, and we define this probability in the \textit{ex-ante} sense, i.e., considering randomness both in the candidate placement and fading realization. An outage may happen because either there is no feasible route for communication, i.e., $P_s> P_{\text{max}}$ and $v_i > P_{\text{max}},\,\forall i\in\mathcal{N}$, or due to the auction mechanism itself. This leads to the following definition:
\begin{definition}\label{def:min_op}
    A mechanism is said to \textbf{minimize the outage probability} if an outage occurs only when the quality of communication channels precludes data offloading, via any of the candidates or via the source itself, within the time frame, $T$. Formally, this is expressed as $\left(\left(\rho(\mathcal{V})=  0 \right)\land (P_s > P_{\text{max}}) \right)\implies \left(\min_{i\in\mathcal{N}} v_i > P_{\text{max}}\right)$. The corresponding \textbf{minimum outage probability} is denoted as $p_{\text{out}}^*$.
\end{definition}
Importantly, when using a mechanism that minimizes the outage probability, communication occurs whenever feasible. 

\subsubsection{Energy Consumption}
The total energy consumption in the system is given by the sum of the source's energy consumption, $T P_{\text{tot}}$, and the possible net energy harvested by the winning candidate, $\tilde{\alpha}_{\rho}T(P_{\text{tot}} - v_{\rho})$, where the dependence of $\rho$ on $\mathcal{V}$ has been dropped for readability in the subscript. Due to the inherent inefficiency of WPT, the candidate harvests much less energy than the source uses. Thus, when discussing energy efficiency, we focus on the source's energy consumption. In out analysis presented in Sections~\ref{sec:vickrey} and \ref{sec:myerson}, we discuss energy consumption in the \textit{ex-post} sense, i.e., considering the the final outcome of the protocol for specific realizations of candidate placement and channel fading. In Section~\ref{sec:results}, numerical results show energy consumption in the \textit{ex-ante} sense, i.e., as expected or average values.

\vspace{-0.1in}
\subsection{A Cooperative Baseline}
To provide a baseline of comparison, we briefly consider the scenario where the source and candidates work cooperatively, so that rather than seeking to maximize their utilities, the candidates are content as long as they are sent enough power to cover the cost of relaying the data on to the AP. In this scenario, the candidates all report $v_i$, and the source relays via the best candidate, $i^* = \argmin_{i} v_i$, using transmit corresponding power $v_{i^*}$ (or if $i^* = 0$ and $P_s > P_{\text{max}}$, it would know that communication is impossible and not transmit at all). Thus, the source could directly minimize its energy consumption.

The cooperative baseline has several nice properties. Clearly, it achieves the minimum outage probability, $p_\text{out}^*$, while also minimizing the source's energy expenditure. Furthermore, it minimizes the system-wide energy expenditure and maximizes the social welfare, that is, the sum of all utilities, as can easily be shown.

Under incentive-compatible mechanisms, candidates also report their valuations, $v_i$, but in contrast, the structure of the mechanism prevents the source from minimizing energy usage in this same way, as candidates only bid their true valuations due to the potential for a positive net energy gain. Fundamentally, it is the lack of information on $H_i$ that leads to the difference in outcomes between the cooperative and non-cooperative scenarios. If the source knows $H_i,\,\forall i \in \mathcal{N}$ beforehand in the non-cooperative case, it can make a take-it-or-leave-it offer to the best candidate, which the candidate would accept since it still has non-negative utility. Thus, the source could achieve the minimal energy consumption even in a non-cooperative setting if it knew $H_i$.

We next evaluate the system under a Vickrey-auction-based protocol. As we shall see, the lack of cooperation leads to some performance loss, but as the number of candidates increases, the results of the auction approach those of the cooperative baseline, as verified in Section~\ref{sec:results}.

\section{Vickrey WPT Auctions (VWA)}\label{sec:vickrey}
This section considers system performance when the protocol described in Section~\ref{sec:auction} employs a Vickrey auction, and we refer to this protocol as the Vickrey WPT Auction (VWA). We first present the Vickrey mechanism, then analyze the outage probability, mathematically characterizing it as a function of the number of candidates. Finally, we discuss energy consumption, which motivates the Myerson mechanism introduced in Section~\ref{sec:myerson}.

\vspace{-0.1in}
\subsection{The Vickrey Mechanism}
The standard Vickrey auction~\cite{vickrey1961JoF} is a classical incentive-compatible mechanism. For a reverse auction, it is defined as:
\begin{equation}\label{eq:vickrey_mechanism}
\begin{split}
    \rho^\text{V}(\mathcal{V}) &= i^* = \argmin_{i\in\mathcal{N}\cup \{0\}} v_i\\
    X^\text{V}(\mathcal{V}) &= \begin{cases}
        v_{\hat{i}} & \text{If } \rho^V(\mathcal{V}) \neq 0\\
        v_0 & \text{If } \rho^V(\mathcal{V}) = 0\\
    \end{cases},
\end{split}
\vspace{-0.1in}
\end{equation}
where $\hat{i} = \argmin_{i\in\mathcal{N}\cup \{0\}\setminus i^*} v_i$
In words, the participant, either candidate or source, with the lowest valuation wins\footnote{Valuations are dependent on communication channel powers, which are continuous random variables, so that we will almost surely never have two valuations equal to each other, and we omit tie-breaking rules for brevity.}. If the winner is a candidate, then the source pays the amount of the second-lowest valuation, while if the winner is the source itself, it ``pays" itself its own valuation (no transfer is made).
\vspace{-0.1in}
\subsection{Outage Probability}
Recall from Section~\ref{sec:auction} that if data transfer is not complete within the time frame $T$, we say an outage has occurred, and let $p_\text{out}^V$ denote the outage probability under the VWA. We now prove that the VWA minimizes the outage probability (Definition~\ref{def:min_op}) before mathematically characterizing this minimum outage probability.

\begin{lemma}\label{lma:vick_min_op}
    The VWA minimizes the outage probability, as described in Definition~\ref{def:min_op}.
\end{lemma}
\begin{proof}
    Recall that if $P_s > P_{\text{max}}$, we have $v_0 = P_{\text{max}}$, and suppose that $P_s > P_{\text{max}}$ and $\rho^V(\mathcal{V}) = 0$. By the definition of $\rho^V(\mathcal{V})$in Eq.~(\ref{eq:vickrey_mechanism}), this implies $v_0 = P_{\text{max}} < \min_{i\in\mathcal{N}} v_i$, which satisfies the condition stated in Definition~\ref{def:min_op}.
\end{proof}
\textbf{As a consequence of Lemma~\ref{lma:vick_min_op}, we have $p_{\text{out}}^V = p_{\text{out}}^*$, and the VWA achieves the same outage probability as the cooperative baseline.} We next mathematically characterize this minimum outage probability.

From the discussion of the valuations in Section~\ref{sec:auction}, we see that an outage occurs if and only if $\min_{i} v_i > P_{\text{max}}$ and $P_{s} > P_{\text{max}}$. For the source, this implies $P_{\max} < \zeta/H_s $, while for each individual candidate, we must have $P_{\text{max}} <  \zeta (1/(\alpha A_rH_{s,i}H_i) + 1/H_{s,i})$. Recall from Section~\ref{sec:modeling} that $H_l = H_{l,\text{PL}} H_{l,\text{f}}$ for $l\in\{s,i,si\}$, where the path loss component $H_{l,\text{PL}}$ is determined by the distance between the two nodes on the link. Let $F_{H_{\text{f}}^{\text{NLOS}}}(h_{\text{f}})$ denote the CDF of the fading in the NLOS case. For a given source location $q_s$, we can then calculate the outage probability at the source as follows:

\begin{equation*}
    \begin{split}
        p_{\text{out},s} &= \text{Pr}\left(P_{\max} < \zeta/H_s\right) = F_{H_{\text{f}}^\text{NLOS}}\left(\zeta/(H_{s,\text{PL}}P_{\text{max}})\right).
    \end{split}
\end{equation*}

The event that the individual candidates experience an outage, i.e., cannot feasibly act as relays, is independently and identically distributed (i.i.d.). To see this, first note that channel measurement campaigns indicate that small-scale fading spatially decorrelates within a few wavelengths, which for high frequencies means within a few centimeters~\cite{Samimi2016VTC}, and we treat fading as independent across channels. We also model the candidates' placement as independent, so that the path loss components of the channel powers are independent as well.

Let $H_c$ and $H_{s,c}$ denote the candidate-to-AP and the source-to-candidate channel power, respectively, for an arbitrary candidate. Further, let $f_{Q}(q)$ denote the PDF of the spatial distribution of the generic candidate, and let $f_{H_{\text{f}}^{\text{LOS}}}(h_{\text{f}})$ and $F_{H_{\text{f}}^{\text{LOS}}}(h_{\text{f}})$ denote the PDF and CDF of the LOS fading, respectively. Then integrating over possible candidate placements and channels realizations, we calculate the outage probability for an arbitrary candidate as:
\vspace{-0.08in}
\begin{equation}\label{eq:candidate_op}
    \begin{split}
        p_{\text{out,c}} &= \text{Pr} \left(P_{\text{max}} <  \zeta (1/(\alpha A_r H_{s,c}H_c) + 1/H_{s,c})\right) \\
        & = \hspace{-3pt} \int_{q\in \mathcal{Q}} \hspace{-5pt} f_{Q}(q)\hspace{-5pt} \int_{h= 0}^{\infty}  f_{H_\text{f}^{\text{LOS}}}(h) F_{H_\text{f}^\text{LOS}}\left(g(q, h) \right)\, dh dq,
    \end{split}
\end{equation}
where $\mathcal{Q}$ is the region in space such that the point $q\in \mathcal{Q}$ has a LOS channel to both the AP and source, and where
\begin{equation*}
    g(q, h)  = \frac{1}{\zeta}\left(\frac{P_\text{max}}{\alpha A_r H_{s,i,\text{PL}}(q)\, H_{i,\text{PL}}(q)\,h} + \frac{P_\text{max}}{H_{s,c,\text{PL}}(q)} \right)^{-1}\hspace{-0.05in}.
\end{equation*}
Here, $H_{c, \text{PL}}(q) = K^{\text{LOS}}||q||_2^{-\eta^{\text{LOS}}}$ and $H_{s,c, \text{PL}}(q) = K^{\text{LOS}}||q - q_s||_2^{-\eta^{\text{LOS}}}$ are the path loss component of the candidate-to-AP and source-to-candidate channels, respectively, given the candidate location, $q$. The overall outage probability is then given by:
\begin{equation}\label{eq:vick_p_out}
    p_{\text{out}}^* =p_\text{out}^V = p_{\text{out},s} {p_{\text{out,c}}}^n .
\end{equation}

The above expression shows that the outage probability decays exponentially as the number of candidates increases, with the base, $p_{\text{out,c}}$, dependent on the spatial distribution of the candidates, among other factors. \textbf{In summary, we have shown that the VWA minimizes the outage probability (see Definition~\ref{def:min_op}), we have mathematically characterized this minimum outage probability, and we have shown that the outage probability decays exponentially as the number of candidates increases.}

\vspace{-0.1in}
\subsection{Energy Consumption}
We next evaluate the system-wide energy consumption, assuming that communication is possible. If $i^* = 0$, then the energy consumption is simply $T v_{i^*} = T v_0 =T P_s$. Importantly, the energy consumption in this case is identical for the cooperative baseline and the VWA. When $i^* \neq 0$, however, the VWA results in additional energy usage. The total energy consumption is given by $Tv_{\hat{i}} - T\tilde{\alpha}_{i^*}(v_{\hat{i}} - v_{i^*})$, where the first term captures the source's energy consumption, and the second term accounts for the net energy harvested by the candidate. Comparing with the cooperative baseline, for which the energy consumption is given by $T v_{i^*}$, the difference in energy is
\vspace{-0.05in}
\begin{equation*}
    \Delta \mathcal{E} = (1-\tilde{\alpha}_{i^*})T(v_{\hat{i}} - v_{i^*}).
    \vspace{-0.05in}
\end{equation*}

The above expression accounts for the fact that the source uses $(v_{\hat{i}} - v_{i^*})$ more power, which is only slightly offset by the net $\tilde{\alpha}_{i^*}(v_{\hat{i}} - v_{i^*})$ power harvested at the relay, since in general $\tilde{\alpha}_{i^*}$ is orders of magnitude less than 1. \textbf{Thus, the VWA is less energy efficient than the cooperative baseline.}

The expression above indicates two approaches for improving energy efficiency. First, improving WPT efficiency $\tilde{\alpha}_i$ through, e.g., hardware improvements, reduces the cost of energy transfer. Second, reducing the gap between $v_{\hat{i}}$ and $v_{i^*}$ will also bring the performance of the VWA more closely in line with cooperative baseline. This could be achieved by increasing the number of candidates, as illustrated in Section~\ref{sec:results}. Alternatively, the source may employ a different type of auction which, on average, brings its total TX power, $X(\mathcal{V}) = P_{\text{tot}}$, closer to the minimum, $v_{i^*}$, while maintaining incentive compatibility. This possibility is explored next.

\section{Myerson WPT Auctions (MWA)}\label{sec:myerson}
In this section, we consider an alternative protocol based on a Myerson auction, which reduces the source's energy usage, and we refer to this protocol as the Myerson WPT Auction (MWA). Myerson auctions \cite{Myerson1981MOR} are known to maximize the auctioneer's utility, so that on average, they will minimize the source's energy consumption. These auctions are defined within a Bayesian setting, so the source must know the probability distribution from which the candidates' valuations are drawn. For the MWA, this amounts to knowing the distributions over $H_i,\,\forall i \in \mathcal{N}$. With this additional information, it's possible to construct a mechanism that, on average, improves the energy efficiency of the system.

The general Myerson mechanism can be computationally cumbersome, but when a set of regularity conditions are met, the mechanism becomes easily computable. In the rest of the section, \textbf{we first extend classical results on regular Myerson auctions to the reverse auction setting and prove the regularity of the MWA for lognormal and Rayleigh fading. We then mathematically characterize the outage probability and show the improved energy efficiency.}

\vspace{-0.1in}
\subsection{Reverse Myerson Auctions}
In this section, we develop a novel extension of results for Myerson auctions. Myerson auctions resemble Vickrey auctions, but rather than comparing valuations directly, \textit{virtual valuations} are used to determine the winner and payment. In general, determination of virtual valuations requires a potentially complex ``ironing" procedure~\cite{Myerson1981MOR}. However, if the distributions over the valuation satisfy a regularity condition, the mechanism simplifies greatly. After first presenting this established regularity condition and associated mechanism for forward auctions, we extend the Myerson framework by mathematically proving the equivalent regularity condition and mechanism for reverse auctions.  

 Let $F_{V_i}(v_i)$ and $f_{V_i}(v_i)$ denote the CDF and PDF, respectively, of a generic valuation, $v_i$, with support $\Theta_i$. We now present the standard regularity definition.

\begin{definition}\label{def:reg}
      (Forward Regularity) Let $c_i^f(v_i) = v_i - \frac{1-F_{V_i}(v_i)}{f_{V_i}(v_i)}$. A forward Myerson auction is \textbf{regular} if $d c_i^f(v_i)/d v_i \geq 0,\,\forall v_i \in \Theta_i, \forall i \in\mathcal{N}$. Furthermore, $c_i^f(v_i)$ is the \textbf{forward auction virtual valuation}.
\end{definition}
Note that the virtual valuations are defined only for the bidders. However, for ease of notation, we define $c_0^f(v_0) = v_0$. We may now restate Myerson's result:
\begin{theorem}\label{thrm:forward_myerson}
    For regular forward Myerson auctions, the seller-optimal mechanism is given by:
    \begin{equation}\label{eq:myerson_forward_mechanism}
        \begin{split}
            \rho^{f,M}(\mathcal{V}) &= \max_{i}\; \left\{c_i^f(v_i)\right\}_{0=1}^n\\
            X^{f,M}(\mathcal{V}) & = \begin{cases}
                \omega^{f,M} & \text{If } \rho^{f,M}(\mathcal{V}) \neq 0\\
                0 & \text{If } \rho^{f,M}(\mathcal{V}) = 0
            \end{cases}
        \end{split}
    \end{equation}
    with $\omega^{f,M} = \inf \left\{s | c_{\rho^{f,M}}^f(s) > \max \{c_i^f(v_i)\}_{i=0}^n\right\}$.
\end{theorem}
\begin{proof}
    See \cite{Myerson1981MOR}.
\end{proof}
Thus, the winner of the forward auction is the bidder with the highest \textit{virtual} valuation, and their payment corresponds to the virtual valuation of the second-highest bidder. The above result holds for forward auctions, but our WPT-enabled relaying scenario is best modeled as a reverse auction, as the bidders are offering to provide a service at a cost. We next extend this result to the reverse auction:

\begin{definition}\label{def:reverse_reg}
      (Reverse Regularity) Let $c_i(v_i) = v_i + \frac{F_{V_i}(v_i)}{f_{V_i}(v_i)}$. A reverse Myerson auction is \textbf{regular} if $d c_i(v_i)/d v_i \geq 0,\,\forall v_i \in \Theta_i,\forall i \in\mathcal{N}$ Furthermore, $c_i$ is the \textbf{reverse auction virtual valuation}.
\end{definition}

For ease of notation, we again set $c_0(v_0) = v_0$, and we now extend Theorem~\ref{thrm:forward_myerson}.
\begin{corollary}\label{crl:reverse_acutions}
   For a regular reverse Myerson auction the buyer-optimal mechanism for reverse auctions is given by:
    \begin{equation}\label{eq:myerson_reverse_mechanism}
        \begin{split}
            \rho^M(\mathcal{V}) &= \min_{i}\quad \{c_i(v_i)\}_{i=0}^n\\
            X^M(\mathcal{V}) & = \begin{cases}
                \omega^M & \text{If } \rho^M(\mathcal{V}) \neq 0\\
                0 & \text{If } \rho^M(\mathcal{V}) = 0
            \end{cases},
        \end{split}
    \end{equation}
    with $\omega^M = \sup \left\{s | c_{\rho^M}(s) < \min \left\{c_i(v_i)\right\}_{i=0}^n\right\}$.
\end{corollary}
\begin{proof} Our proof consists of two parts. First, we show that Definitions~\ref{def:reg} and \ref{def:reverse_reg} are equivalent. Then, we show that the two mechanisms are equivalent as well. We first cast the reverse auction as a forward auction by negating the valuations. Let $v_i^f = -v_i$ denote the equivalent forward auction valuations. Then we have 
    \begin{equation}\label{eq:neg_of_VV}
    \begin{split}
     c_i^f(v_i^f) &= c_i^f(-v_i) = -v_i - \frac{1 - F_{V_i^f}(-v_i)}{f_{V_i^f}(-v_i)} \\
     & = -v_i -  \frac{F_{V_i}(v_i)}{f_{V_i}(v_i)} = -c_i(v_i).
    \end{split}
    \end{equation}
    Thus, $c_i(v_i)$ corresponds to the negation of the forward auction virtual valuations. Furthermore, we note that 
    \begin{equation*}
        \frac{d c_i^f(v_i^f)}{d v_i^f} = \frac{(-1)d c_i(v_i)}{(-1)d v_i} = \frac{d c_i(v_i)}{d v_i}.
    \end{equation*}
Therefore, the valuations of the reverse auction satisfy the condition in Definition~\ref{def:reverse_reg} if and only if the valuations of the equivalent forward auction satisfy Definition~\ref{def:reg}.

We must also show that the reverse auction and forward auction mechanisms are equivalent, i.e., $\rho^{M}(\mathcal{V}) = \rho^{f,M}(-\mathcal{V})$ and $X^M(\mathcal{V}) = -X^{f,M}(-\mathcal{V})$, where $-\mathcal{V}$ is the set of valuations in the equivalent forward auction. This follows directly from the negation of the virtual valuations shown in Eq.~(\ref{eq:neg_of_VV}) and the definitions of the mechanisms in Eq.~(\ref{eq:myerson_forward_mechanism}) and Eq.~(\ref{eq:myerson_reverse_mechanism}).
\end{proof}
\vspace{-0.1in}
\textbf{By proving the Myerson regularity condition and mechanism structure for reverse auctions, we have derived conditions under which the Myerson pricing mechanism is easily implementable.} We next show that under lognormal fading and Rayleigh fading, the MWA satisfies the reverse regularity condition in Definition~\ref{def:reverse_reg}.

\vspace{-0.15in}
\subsection{Regularity of WPT Auctions}
We now show the regularity of the MWA under lognormal and Rayleigh fading. For tractable analysis, we assume that the source has partial information about the candidate-to-AP downlink channel power, $H_i$. Specifically, we assume the locations of the candidates and path loss parameters are available to the source, so that the path loss component of the channel, $H_{i,\text{PL}}$, is known. We also assume that the parameters needed to model fading, specified below, are known\footnote{Candidate locations may be found using, e.g., integrated sensing and communication techniques~\cite{Lu2024IoTJ}. The LOS channel path loss and fading parameters may be estimated from the source's interaction with the candidates over the source-to-candidate channel, which is also LOS.}. However, the fading realizations, $H_{i,\text{f}},\,\forall i \in \mathcal{N}$ are unknown, and they are treated as i.i.d random variables with a known PDF and CDF of $f_{H_f^{\text{LOS}}}(h)$ and $F_{H_f^{\text{LOS}}}(h)$, respectively.

To simplify analysis, recall that $v_i = P_i/\tilde{\alpha}_i + P_{s,i},\,\forall i\in\mathcal{N}$. Given our information assumptions, $P_{s,i}$ is constant while $P_i$ is a random variable due to its dependence on the unknown fading power $H_{i,\text{f}}$. Letting $\tilde{v}_i =P_i/\tilde{\alpha}_i = C_i/H_{i,\text{f}}$, with
$C_i = \zeta/(H_{i,\text{PL}}\tilde{\alpha_i})$, we have $f_{V_i}\left( v_i\right) = f_{\tilde{V}_i}\left( v_i-P_{s,i}\right)$. Further, let $\tilde{c}_i(\tilde{v}_i) = \tilde{v}_i + F_{\tilde{V_i}}(\tilde{v}_i)/f_{\tilde{V_i}}(\tilde{v}_i)$, and note that $d\tilde{c}_i(\tilde{v}_i)/d \tilde{v}_i = d c_i(v_i)/d v_i$. Thus, we may equivalently show that the auction with valuations $\tilde{v}_i$ is regular. We now derive the cumulative distribution function  of $\tilde{v}_i$:

\vspace{-0.1in}
\begin{equation}\label{eq:general_vtilde_dist}
\begin{split}
    F_{\tilde{V}_i}(\tilde{v}_i) &= \text{Pr}\left(\tilde{V}_i \leq \tilde{v_i} \right) = \text{Pr}\left(\frac{C_i}{H_{i,\text{f}}} \leq \tilde{v}_i\right) \\
    &=  1 - \text{Pr}\left(H_{i,\text{f}} \leq C_i/\tilde{v}_i\right) =1-F_{H_{f}^\text{LOS}}\left(h_i(\tilde{v}_i)\right),
\end{split}
\vspace{-0.1in}
\end{equation}
where $h_i(\tilde{v}_i) = C_i/\tilde{v}_i$. Noting that $h_i(\tilde{v})$ is monotonic in $\tilde{v}$, the PDF is given by $f_{\tilde{v}_i}(\tilde{v}_i) = f_{H_{f}^{\text{LOS}}}\left(h_i(\tilde{v}_i)\right)C_i/\tilde{v}_i^2$. We are now ready to develop the main results.

\noindent\textbf{Lognormal fading:}
We first consider the case where fading power is modeled as a lognormal random variable. Discussion of lognormal fading generally characterize channel power in the dB domain, i.e., $H_{i,\text{f}}^{\text{dB}} = 10\log_{10}H_{i,\text{f}}$. For the derivations below, it will be more convenient to work with linear power, so that if $H_{i,\text{f}}\sim \text{Lognormal}(0, \sigma_f^2)$, then $H_{i,\text{f}}^{\text{dB}}\sim N(0, (10\sigma_f/\ln(10))^2)$, where $N$ denotes the normal distribution.

\begin{theorem}
    Suppose that for all $i\in\mathcal{N}$, the fading is i.i.d with $H_{i,\text{f}}\sim \text{Lognormal}(0, \sigma_f^2)$. Then the WPT reverse Myerson auction is regular.
\end{theorem}
\begin{proof}
As a general statement, we have
\begin{equation}\label{eq:general_dv}
    \begin{split}
        \frac{d\tilde{c}_i(\tilde{v}_i)}{d\tilde{v}_i} &= 1 +\frac{f_{\tilde{V}_i}(\tilde{v_i})}{f_{\tilde{V}_i}(\tilde{v_i})} - \frac{F_{\tilde{V}_i}(\tilde{v_i})}{f_{\tilde{V}_i}(\tilde{v_i})^2}\frac{df_{\tilde{V}_i}(\tilde{v_i})}{d\tilde{v}_i}\\
    &= 2- \frac{F_{\tilde{V}_i}(\tilde{v_i})}{f_{\tilde{V}_i}(\tilde{v_i})^2}\frac{df_{\tilde{V}_i}(\tilde{v_i})}{d\tilde{v}_i}.
    \end{split}
\end{equation}
For lognormal shadowing, the CDF of the fading is $F_{H_{f}^\text{LOS}}(h) = \Phi(\ln{(h)}/\sigma_{f})$, with $\Phi(\cdot)$ the CDF of the standard normal distribution, while the PDF is $f_{H_{f}^\text{LOS}}(h) = \frac{1}{h\sigma_{f}\sqrt{2\pi}}\text{exp}(-\ln{(h)}/2\sigma_{f}^2)$. Based on this, we have% $ F_{\tilde{V}_i}(v) = 1 -\Phi(\ln{(h(v))}/\sigma_{\text{f}})$, and
\begin{equation*}
    F_{\tilde{V}_i}(\tilde{v}_i) = 1 -\Phi\left(\frac{\ln{(h_i(\tilde{v}_i))}}{\sigma_{\text{f}}}\right)
\end{equation*}
and
\begin{equation*}
    f_{\tilde{V}_i}(\tilde{v}_i) = \frac{d F_{\tilde{V}_i}(\tilde{v}_i)}{d \tilde{v}_i} = \frac{h_i(\tilde{v}_i)}{C_i \sigma_{f}\sqrt{2\pi}}\text{exp}\left(\frac{-\ln{(h_i(\tilde{v}_i))}}{2\sigma_{f}^2}\right).
\end{equation*}  
Differentiating, we get
\begin{equation*}
    \begin{split}
        \frac{d f_{\tilde{V}_i}(\tilde{v}_i)}{d\tilde{v}_i} &= \frac{d f_{\tilde{V}_i}(\tilde{v}_i)}{d h_i(\tilde{v}_i)} \frac{{d h_i(\tilde{v}_i)}}{d\tilde{v}_i}\\
        &= \frac{\text{exp}(-\ln{h_i(\tilde{v}_i)}/(2\sigma_{f}^2))}{C_1 \sigma_{f}\sqrt{2\pi}}\left(1 - \frac{\ln{h_i(\tilde{v}_i)}}{\sigma_{f}^2}\right)\frac{-C_i}{\tilde{v}_i^2}\\
        & = \frac{-\text{exp}(-\ln{h_i(\tilde{v}_i)}/(2\sigma_{f}^2))}{\tilde{v}_i^2 \sigma_{f}\sqrt{2\pi}}\left(1 - \frac{\ln{h_i(\tilde{v}_i)}}{\sigma_{f}^2}\right)\\
        &= \frac{-f_{\tilde{V}_i}(\tilde{v}_i)}{\tilde{v}_i}\left(1 - \frac{\ln{h_i(\tilde{v}_i)}}{\sigma_{f}^2}\right).
    \end{split}
\end{equation*}

We can now expand Eq.~(\ref{eq:general_dv}):
\begin{equation*}
    \begin{split}
         \frac{d\tilde{c}_i(\tilde{v}_i)}{d\tilde{v}_i} &= 2- \frac{F_{\tilde{V}_i}(\tilde{v_i})}{f_{\tilde{V}_i}(\tilde{v_i})^2}\frac{-f_{\tilde{V}_i}(\tilde{v}_i)}{\tilde{v}_i}\,\left(1 - \frac{\ln{h_i(\tilde{v}_i)}}{\sigma_{f}^2}\right)\\
         &= 2 + \frac{F_{\tilde{V}_i}(\tilde{v_i})}{f_{\tilde{V}_i}(\tilde{v_i})\tilde{v}_i}\,\left(1 - \frac{\ln{h_i(\tilde{v}_i)}}{\sigma_{f}^2}\right)\\
         &= 2 + \frac{F_{\tilde{V}_i}(\tilde{v_i})}{f_{H_{\text{f}}}(h_i(\tilde{v_i})) h_i(\tilde{v}_i)}\,\left(1 - \frac{\ln{h_i(\tilde{v}_i)}}{\sigma_{f}^2}\right)\\
         &= 2 + \frac{1 - \Phi\left(\frac{\ln{h_i(\tilde{v}_i)}}{\sigma_{f}}\right)}{\phi\left(\frac{\ln{h_i(\tilde{v}_i)}}{\sigma_{f}}\right) \sigma_{f}^{-1}}\,\left(1 - \frac{\ln{h_i(\tilde{v}_i)}}{\sigma_{f}^2}\right).\\
    \end{split}
\end{equation*}
We then let $\gamma = \ln{(h_i(\tilde{v}_i}))/\sigma_{f}$ to get
\begin{equation*}
     \frac{d\tilde{c}_i(\tilde{v}_i)}{d\tilde{v}_i} = 2 + \frac{1 - \Phi\left(\gamma\right)}{\phi\left(\gamma\right)}\,\left(\sigma_{f} -\gamma\right).
\end{equation*}
where $\phi$ is the PDF of the standard normal distribution. From properties of Mills' ratio, we have $(1-\Phi(\gamma))/\phi(\gamma) < 1/\gamma,\;\forall \gamma >0$ \cite{Gordon1941AMS}, consequently, for all $\tilde{v}_i$, we have $d\tilde{c}_i(\tilde{v}_i)/d \tilde{v}_i > 0$.
\end{proof}

\textbf{The above result shows that under lognormal fading, the mechanism described in (\ref{eq:myerson_reverse_mechanism}) minimizes the source's expected energy consumption.} Furthermore, as a consequence of the monotonicity of $c_i(v_i)$ implied by regularity, the optimal transmit power $X^M(\mathcal{V})$ can be found  with logarithmic time complexity using a simple bisection method.

\noindent\textbf{Rayleigh fading:}
Rayleigh fading is another common fading model found in the literature. Importantly, Rayleigh fading refers to the distribution of the channel gain, the square root of the power, and if the channel gain is Rayleigh distributed with distribution parameter $\psi_{f}$, the channel power is exponentially distributed with average power $2\psi_f^2$ \cite{GoldsmithWC}. We now prove the regularity of the auction assuming this model.
\begin{theorem}
    Suppose that for all $i\in\mathcal{N}$, the fading is i.i.d with $H_{i,\text{f}}\sim \text{Exp}(1/(2\psi_{\text{f}}^2)$. Then the WPT reverse Myerson auction is regular.
\end{theorem}
\begin{proof}
    From Eq.~(\ref{eq:general_vtilde_dist}) and the properties of the exponential distribution, we have 
\begin{equation*}
    \begin{split}
    F_{\tilde{V}_i}(\tilde{v}_i) &= 1 - (1-e^{-h_i(\tilde{v}_i)/(2\psi_{\text{f}}^2)}) = \text{exp}\left(\frac{-C_i}{2\tilde{v}_i\psi_{\text{f}}^2}\right)\\
    f_{\tilde{V}}(\tilde{v}_i) &= \frac{C_i}{2\tilde{v}_i^2\psi_{\text{f}}^2}\text{exp}\left(\frac{-C_i}{2\tilde{v}_i\psi_{\text{f}}^2}\right).
    \end{split}
\end{equation*}
 We can now write $\tilde{c}_i(\tilde{v}_i) = \tilde{v}_i + 2\tilde{v}_i^2 \psi_{\text{f}}^2/C_i$,
% \begin{equation*}
%     \tilde{c}_i(\tilde{v}_i) = \tilde{v}_i + 2\tilde{v}_i^2 \psi_{\text{f}}^2/C_i
% \end{equation*}
which is monotonically increasing in $\tilde{v}_i$ for all $\tilde{v}_i \geq 0$.
\end{proof}
\textbf{The result presented above indicates that, with Rayleigh fading, the mechanism given in (\ref{eq:myerson_reverse_mechanism}) minimizes the source's average energy usage.} We note that the expression $\tilde{c}_i(\tilde{v}_i)$ is quadratic, so that finding the payment, $X^M(\mathcal{V})$, using the Myerson mechanism can be done in constant time.

\noindent \textbf{Remark - Rician Fading:} In addition to lognormal and Rayleigh fading, Rician fading is another important fading model found in the literature. Although Rayleigh fading can be seen as a Rician fading with a K-factor of 0 \cite{GoldsmithWC},  \textbf{it is easy to produce counterexamples that illustrate that a general Rician distribution is not regular}, and thus more complex ``ironing" techniques would be needed for auctions in these scenarios \cite{Myerson1981MOR}. This makes Myerson auctions under Rician fading more computationally challenging compared to lognormal and Rayleigh fading, and we discuss this point further in Section~\ref{sec:conclusion}.

\begin{figure*}
    \centering
    \includegraphics[width=7.14in, trim={0.1in, 0.05in, 0.1in, 0.45in}, clip]{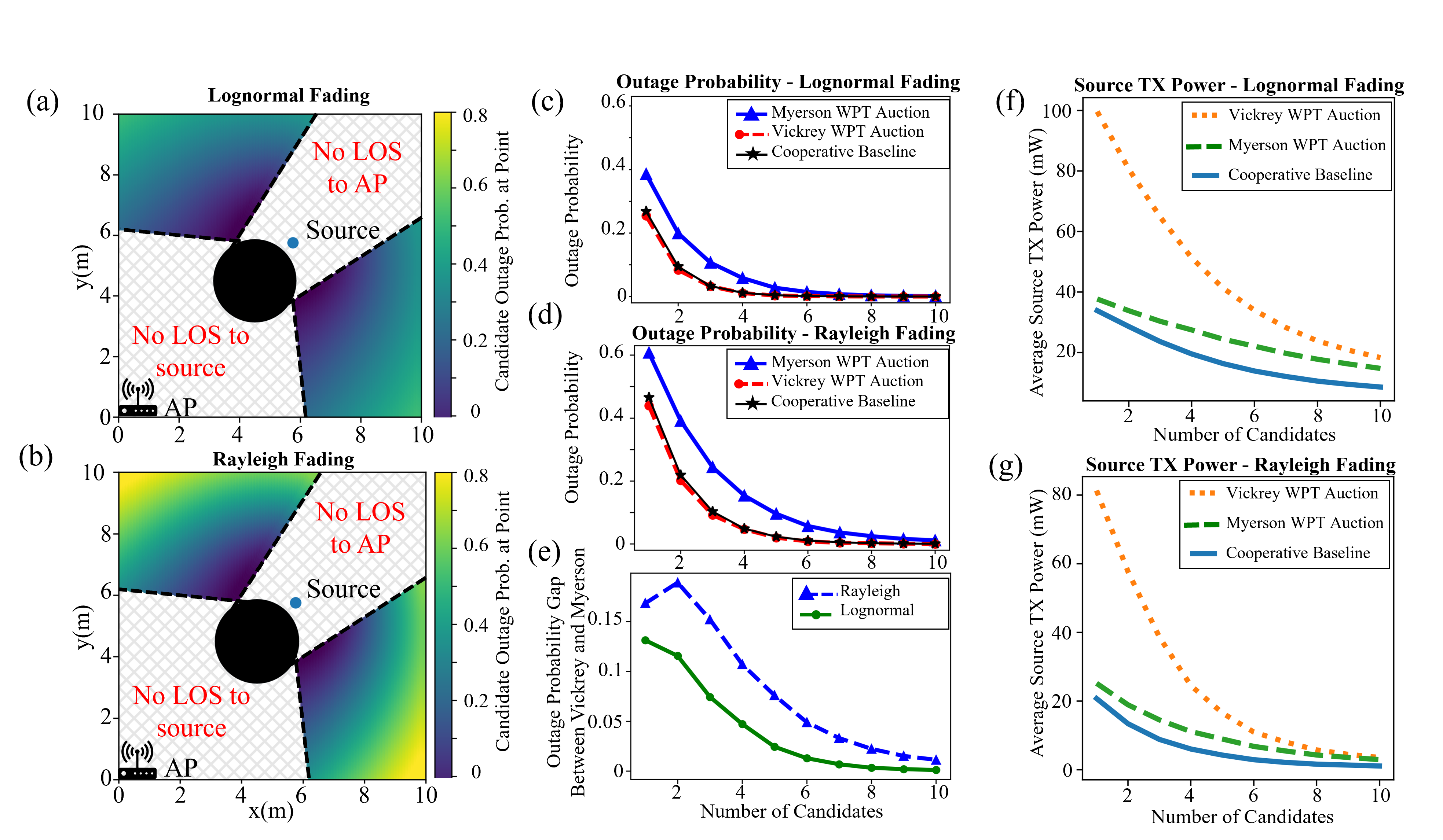}
    \vspace{-0.2in}
    \caption[short]{Numerical results for a sample blockage scenario. (a) and (b): The probability that a candidate at each point could act as a relay in the cooperative case for (a) lognormal and (b) Rayleigh fading, i.e., $p_{\text{out},c}$ from Eq.~(\ref{eq:candidate_op}) conditioned on candidate placement. Areas which do not have a LOS path to both the AP and the source are crosshatched. (c) and (d): The outage probability of the MWA, VWA, and cooperative baseline for (c) lognormal and (d) Rayleigh fading. The cooperative baseline outage probability is calculated using the expression for the minimal outage probability in Eq.~(\ref{eq:vick_p_out}), while the MWA and VWA outage probabilities are determined empirically from the numerical experiments.  The outage probability under the MWA is greater than it is for the VWA, which matches the minimum outage probability. For both the VWA and MWA, the outage probability decays exponentially. (e): The outage probability gap between the Vickrey and Myerson auction. This gap decays exponentially in the number of candidates. (f) and (g): The average source transmit power under the VWA and MWA, as well as for the cooperative baseline introduced in Section~\ref{sec:auction}, for (f) lognormal and (g) Rayleigh fading. Energy usage for the MWA is less than it is under the VWA, and energy efficiency converges as the number of candidates increases. See color PDF for optimal viewing.}
\label{fig:fixed_exp_setup}
\vspace{-0.2in}
\end{figure*}

\vspace{-0.1in}
\subsection{Outage Probability}
\textbf{While the MWA promises greater energy efficiency on average compared to the VWA, it does not minimize the outage probability.} To illustrate, consider a simple scenario consisting of a single candidate and the source, with $P_{s} > P_{\text{max}}$ and $v_1 < P_{\text{max}}  = v_0 < c_1(v_1)$. The VWA results in the candidate winning and the source using total power $P_{\text{max}}$. However, for the MWA, the source wins but cannot communicate. Thus, while the MWA requires less energy on average, it does so at the cost of increased outage probability. We next characterize the extent to which the MWA results in additional outage probability.

For the MWA, we decompose the outage probability into two components. First, there is $p_{\text{out}}^*$, which is the minimum outage probability, corresponding the event that $P_s > P_{\text{max}}$ and $ v_i > P_{\text{max}}$ for all $i>0$. For the MWA, we have an additional outage probability, $p_{\text{out}}^{\text{M,g}}$, corresponding to realizations of the valuations which satisfy three conditions: that $P_s > P_{\text{max}}$, that $\rho(\mathcal{V}) = 0$, and that there exists at least one $i$ such that $v_i \leq P_\text{max}$. Thus, the outage probability under the MWA, $p_\text{out}^\text{m} = p_\text{out}^* + p_{\text{out}}^{\text{M,g}}$. We have already mathematically characterized $p_\text{out}^*$ in Eq. (\ref{eq:vick_p_out}), and we now characterize the gap, $p_{\text{out}}^{\text{M,g}}$:
% \vspace{-0.1in}
\begin{equation*}
\begin{split}
    p_{\text{out}}^{\text{M,g}}\hspace{-0.02in} =& \text{Pr}\left(\left(P_s > P_{\text{max}}\right)  \bigcap \left(\cap_{i=1}^n c_i(v_i) > v_0  \right) \bigcap \left( \cup_{i=1}^n v_i \leq v_0\right) \right)\\
    =&\text{Pr}\left(P_s > P_{\text{max}} \right) \\
    &\times \text{Pr}\left( \left( \cap_{i=1}^n c_i(v_i) > v_0 \right) \bigcap \left(\cup_{i=1}^n v_i \leq v_0 \right)|P_s > P_{\text{max}}\right)\\
    & = p_{\text{out},s} \prod_{i=1}^n \left( F_{V_i}(P_{\text{max}}) - F_{V_i}(\bar{\nu}_{i}) \right),
\end{split}
\end{equation*}
with $\bar{\nu}_{i} = c_i^{-1}(P_{\text{max}})$, where, assuming regularity, $c_i^{-1}$ is the well-defined inverse of the reverse auction virtual valuation function, $c_i(v_i)$. 

Three circumstances reduce or eliminate the increased risk of outage. First, as long as a candidate wins the MWA, communication will occur, since the payment $X^M(\mathcal{V})$ will satisfy $X^M(\mathcal{V}) < c_{\rho^M}(X^M(\mathcal{V})) \leq v_0\leq P_{\text{max}}$. Secondly, if $P_{s} \leq P_{\text{max}}$, then communication will always occur when possible, since even if the source wins, it is able to successfully transmit the data in time. Finally, as the number of candidates increases, the additional outage probability decreases. More specifically, if we again assume that candidates are independently and identically distributed (so that valuations, $v_i\,\,\forall i \in \mathcal{N}$ are i.i.d.), we have 
\begin{equation*}
    p_{\text{out}}^{\text{M,g}} = P_{\text{out},s} \left(F_{V_c}(P_{\text{max}}) - F_{V_c}(\bar{\nu}_{c})\right)^n,
\end{equation*}
where the candidate-identifying subscript $i$ has been replaced by $c$, denoting a generic candidate. This expression shows that the gap, $p_{\text{out}}^{\text{M,g}}$, between the minimum outage probability and the outage probability of the MWA decays exponentially in $n$. \textbf{Thus, we have shown that there is gap in outage probability between the cooperative baseline and the MWA, and that this gap decreases exponentially as the number of candidates increases.}

\vspace{-0.1in}
\subsection{Energy Consumption}
While the Myerson mechanism may result in a larger outage probability, it also reduces the average energy consumption. To build intuition as to why, consider a simple scenario with a single candidate, and assume that $v_1  < c_1(v_1) < v_0$. In the case of the VWA, the total source transmit power is $P_\text{tot} = v_0$, while for the MWA, we have $P_\text{tot} = c_1^{-1}(v_0) < v_0$. As numerical results in Section~\ref{sec:results} show, this difference can be significant. 

\textbf{In summary, the MWA increases energy efficiency but does so at the cost of a greater outage probability.} We next present several numerical results which illustrate the analysis of the preceding sections.

%%%%%%%%%%%%%%%%%%%%%%%%%%%%%%%%%%%%%%%%%%%%%%%%%%%%%%%%%%%%%%%%
%%%%%%%%%%%%%%%%%%%%%%%%%%%%%%%%%%%%%%%%%%%%%%%%%%%%%%%%%%%%%%%%
%%%%%%%%%%%%%%%%%%%%%%%%%%%%%%%%%%%%%%%%%%%%%%%%%%%%%%%%%%%%%%%%
% Results
%%%%%%%%%%%%%%%%%%%%%%%%%%%%%%%%%%%%%%%%%%%%%%%%%%%%%%%%%%%%%%%%
%%%%%%%%%%%%%%%%%%%%%%%%%%%%%%%%%%%%%%%%%%%%%%%%%%%%%%%%%%%%%%%%
%%%%%%%%%%%%%%%%%%%%%%%%%%%%%%%%%%%%%%%%%%%%%%%%%%%%%%%%%%%%%%%%
\section{Numerical Results}\label{sec:results}
We now present numerical studies illustrating the analysis presented in the previous sections. We compare the outage probability, the source energy usage, and the energy harvested using both the VWA and MWA under different fading assumptions. These showcase when one type of auction will be preferred over the other, as well as how the number of candidates, the fading characteristics, and the blockage geometry impact system performance.

Throughout the section, we consider the setup depicted in Fig.~\ref{fig:fixed_exp_setup}, with $q_s = (5.76\,\text{m}, 5.76\,\text{m})$, the AP at the origin, and a large circular blockage occluding the LOS path. The locations of the candidates, $q_i$, are i.i.d. uniformly across the regions where a LOS path exists to both the AP and source, i.e., the non-gray areas in Fig.~\ref{fig:fixed_exp_setup} (a) and (b). We use the following system parameters, drawn from real-world studies: For LOS channels, we set the path loss intercept and exponent to $K^{\text{LOS}} = 0$ and $\eta^{\text{LOS}} = 2.5$, respectively, while for NLOS channels we set $K^{\text{NLOS}} = -25$ and $\eta^{\text{NLOS}} = 5.76$. Both lognormal and Rayleigh fading is considered, and the lognormal fading is parameterized with $\sigma_{\text{f}} = 8.66$ for LOS channels and $\sigma_{\text{f}} = 9.02$ for NLOS channels \cite[Table VI]{Hemadeh2018CST},~\cite{Zhao2013ICC}. We set $\psi_f = 1/\sqrt{2}$ for all Rayleigh fading experiments to ensure that the same total power is used for the lognormal and Rayleigh cases. The receiver noise power is taken as $\sigma^2 = -75\,$dBm, and $P_{\text{max}} = 100\,$mW. We consider a fixed time of $T = 1\,$s and a bandwidth-normalized payload of $D = 8$ bits/Hz, so that the data must be sent at a rate of at least $r = 8\,$ bits/s/Hz. If the maximum achievable rate as defined in Eq.~(\ref{eq:max_rate}) is less than this on any of the channels used for communication, we declare an outage. For WPT, we take the circuity efficiency to be $\alpha=0.2$~\cite{kwiatkowski20202IMS} and the effective aperture as $A_r=0.01\,\text{m}^2$.

Throughout, we show system performance with different numbers of candidates, under different fading models, and in different environments. For each combination of number of candidates, fading model, and environment geometry, we sample 10,000 different possible placements of the $n$ candidates, and for each candidate we sample fading variables from the corresponding distribution. For each sampled problem instance, we calculate the cooperative baseline and the results of both the VWA and MWA. We report the average of these batches of 10,000 problem instances.

\begin{figure}
    \centering
    \includegraphics[width=0.8\linewidth, trim = {0.05in, 0.1in, 0.1in, 0.325in}, clip]{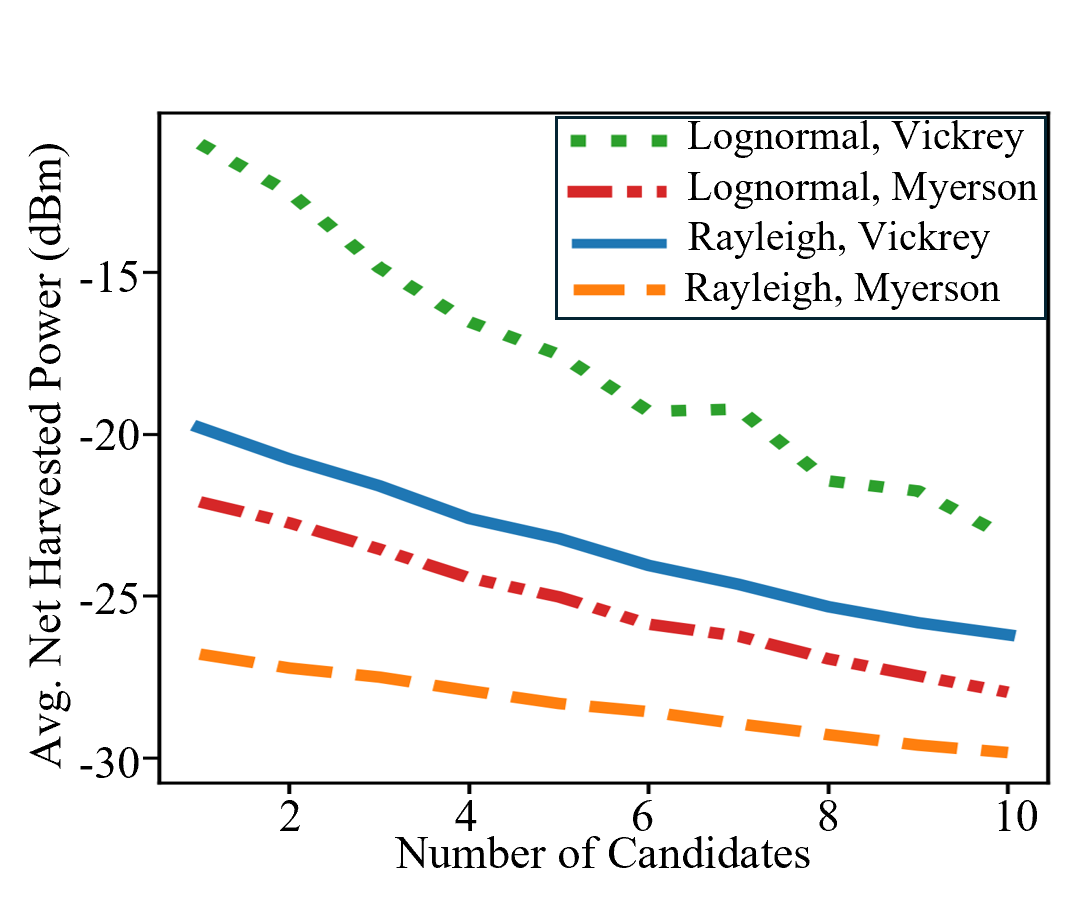}
    \caption{Average net energy harvested for the scenario depicted in Fig.~\ref{fig:fixed_exp_setup}.}
    \label{fig:fixed_exp_harvest}
    \vspace{-0.2in}
\end{figure}

\begin{figure*}
    \centering
    \includegraphics[width = 7.14in, trim={0.1in, 0.05in, 0.1in, 0.05in}, clip]{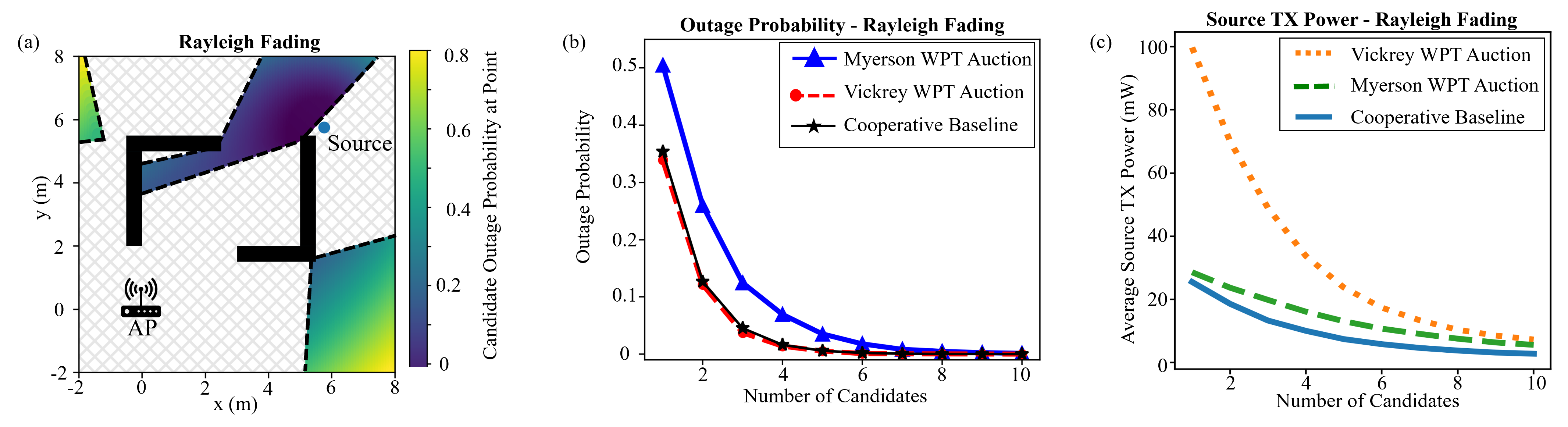}
    \vspace{-0.25in}
    \caption{Numerical results for the blockage scenario shown in (a) under Rayleigh fading. (a): The probability that a candidate at each point could act as a relay in the cooperative case,i.e., $p_{\text{out},c}$ from Eq.~(\ref{eq:candidate_op}) conditioned on candidate placement. Areas which do not have a LOS path to both the AP and the source are crosshatched. (b): The empirical outage probability of the MWA and VWA, and the analytical outage probability of cooperative baseline, calculated using Eq.~(\ref{eq:vick_p_out}). (c): The average source transmit power under the VWA and MWA, as well as for the cooperative baseline introduced in Section~\ref{sec:auction}. See color PDF for optimal viewing.}
    \label{fig:config_3_results}
    % \vspace{-0.15in}
\end{figure*}

%%Outage probability
\vspace{-0.1in}
\subsection{Outage Probability} We first study the outage probability across auction types, looking at the impact of the number of candidates and the fading characteristics. Averaging over the possible positions of the candidates as well as the possible realizations of the channel fading power, we find that the outage probability for an arbitrary agent, mathematically characterized in Eq.~(\ref{eq:candidate_op}), is $p_{\text{out,c}} = 0.35$ for lognormal fading and $p_{\text{out,c}} = 0.47$ for Rayleigh fading (Fig.~\ref{fig:fixed_exp_setup} (b)). Thus, we will see lower outage probabilities under lognormal fading than we see under Rayleigh fading, though we note that for different choices of variance $\sigma_f^2$, this relationship may not hold. These candidate outage probabilities are then used to calculate the minimum outage probability $p_\text{out}^*$ using Eq.~(\ref{eq:vick_p_out}), which is compared to the outage probabilities of the VWA and MWA computed empirically from the simulation experiments. As stated in Lemma~\ref{lma:vick_min_op}, we see that the VWA achieves the minimum outage probability, i.e., $p_\text{out}^* = p_\text{out}^V$, while there is a gap between the minimum outage probability and the outage probability of the MWA. The results also demonstrate the exponential decay predicted in our discussion in Sections~\ref{sec:vickrey}, and Fig.~\ref{fig:fixed_exp_setup} (e) shows that the gap in outage probability between the MWA and minimum outage probability decays exponentially, as discussed in Section~\ref{sec:myerson}. Crucially, both the VWA and MWA make the outage probability very small for large values of $n$.

\vspace{-0.1in}
\subsection{Source Energy Usage}
Figs.~\ref{fig:fixed_exp_setup} (f) and (g) compare the source's average energy consumption under both the VWA and MWA, with the cooperative scenario as a baseline. For fair comparison, we take averages for each of the methods only over cases with successful communication. We see that, as expected, the MWA results in significantly less energy usage than the VWA, and that for large $n$, the energy consumption of both auctions approaches that of the baseline. The decay in energy consumption is particularly pronounced for the VMA. Thus, while the VWA achieves a lower outage probability than the MWA, the MWA is more energy efficient.

\vspace{-0.1in}
\subsection{Energy Harvested}
The average energy harvested by the candidates is plotted in Fig.~\ref{fig:fixed_exp_harvest}. Increasing the number of candidates makes the auction more competitive, and consequently we see the energy harvested decreases. Similarly, the MWA reduces the net energy harvested by the candidates when compared to the VMA, as the MWA allocates more of the net energy savings of the system to the source. As discussed in earlier sections, the net energy harvested by the candidates is several orders of magnitude smaller than the energy used by the source, so that source energy consumption dominates considerations of energy efficiency in the system.

\vspace{-0.1in}
\subsection{Impact of Blockage Geometry}

To better understand the impact of the environment, we consider an additional blockage configuration, shown in Fig.~\ref{fig:config_3_results} (a), under Rayleigh fading. Due to the geometry of the blockages in this environment, the outage probability for an arbitrary candidate is $p_{\text{out,c}} = 0.36$, lower than it was in the blockage environment considered in Fig.~\ref{fig:fixed_exp_setup} under Rayleigh fading, and comparing Fig.~\ref{fig:config_3_results} (b) to Fig.~\ref{fig:fixed_exp_setup} (d), we see that the outage probabilities drop more quickly in this environment than they do in the environment considered previously. Similarly, the average source TX power drops more quickly in this alternate blockage configuration (compare Fig.~\ref{fig:config_3_results} (b) to Fig.~\ref{fig:fixed_exp_setup} (g)). Thus, while the overall trends of rapid decay in outage probability and energy consumption are consistent across environments, certain blockage environments lead to better performance than others.

\begin{figure}
    \centering
    \includegraphics[width=0.8\linewidth, trim={0.1in, 0.1in, 0.1in, 0.1in}]{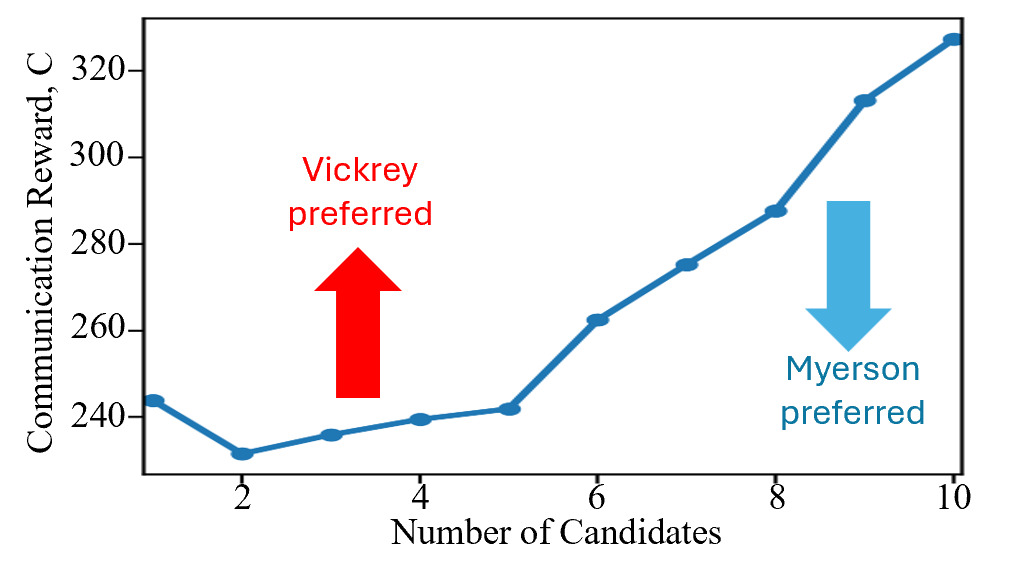}
    \vspace{-0.05in}
    \caption{Proposed design space: The blue line indicates the values of the communication reward, $C$, defined in Eq.~(\ref{eq:utilities}), such that the source's expected utility is equal under the VWA and MWA. Conceptually, $C$ is a penalty for failing to communicate. For values of $C$ below the line, the MWA is preferred. For values of $C$ above the line, the VWA is preferred.}
    \label{fig:fixed_exp_critical_C}
    \vspace{-0.15in}
\end{figure}
\vspace{-0.1in}

\subsection{Comparing Auction Types}
The previous results illustrate the inherent tradeoff between the VWA and MWA. For small $n$, the MWA provides significant energy savings compared to the VWA, but at the cost of a higher Outage probability. As $n$ grows, the energy savings offered by the MWA become less dramatic. However, the outage probability gap, $p_{\text{out}}^{M,g}$, also shrinks.

To better understand when one auction type should be used instead of the other, we compare the average utility, $U_s$, as defined in Eq.~(\ref{eq:utilities}) for both the MWA and VWA under Rayleigh fading in the environment shown in Fig.~\ref{fig:fixed_exp_setup}. Specifically, for each value of $n$, we find the value of $C$ (the reward for successful communication in Eq.~(\ref{eq:utilities})) such that the source is indifferent between the two auctions and plot them in Fig.~\ref{fig:fixed_exp_critical_C}. As $C$ grows larger, greater importance is placed on ensuring communication occurs, i.e., reducing outage probability, and energy efficiency becomes less critical. We therefore see that when ensuring communication is more important (larger values of $C$), the VWA is preferred, as it achieves the minimum outage probability. On the other hand, the MWA is preferred when greater importance is placed on energy efficiency (smaller values of $C$). Qualitatively, we see that as $n$ increases, the MWA becomes preferred for a greater range of values of $C$. This suggests that as the number of candidates increases, the MWA becomes more desirable.

\vspace{-0.1in}
\section{Conclusions and Future Work}\label{sec:conclusion}
This paper presents a novel approach to incentivizing non-cooperative, battery-powered UEs to act as relays through the use of wireless power transfer (WPT) and auction-based protocols. By analyzing the performance of Vickrey and Myerson auctions under realistic channel conditions, we have demonstrated that these mechanisms offer promising solutions for optimizing both energy efficiency and outage probability in non-cooperative relay networks with limited information. Through mathematical analysis and proofs, we have shown that the Vickrey WPT auction (VWA) achieves the same minimum outage probability achieved in the non-cooperative case, while the Myerson WPT auction (MWA) is an energy-efficient alternative at the cost of a possibly greater outage probability. Extensive numerical results corroborate our analytical findings and highlight the critical role of system parameters, such as the number of candidate UEs, the fading environments, and the geometry of the blockages. Overall, the mathematical analysis and numerical results provide guidance on when one auction type is preferred over the other.

Several directions remain open for future work. While the VWA is applicable for any fading model, this paper focused on lognormal and Rayleigh fading when discussing the MWA, as these models make the Myerson mechanism tractable. Extending the MWA to handle non-regular Rician fading models through the application of ``ironing" procedures~\cite{Myerson1981MOR} would expand the applicability of our auction-based approach. Another crucial area of research involves exploring the implications of non-winning candidates still harvesting some power from the source's transmission. Additionally, relaxing the constraint on transmission time would allow for additional optimization over $T$, which adds an interesting element to the analysis of the auctions and makes the framework suitable for instances without hard time constraints. Finally, considering these auctions with a mobile source in the context of communication-aware robotics~\cite{hurst2024unmannedvehicles6gnetworks} would add another interesting layer of complexity.

\bibliographystyle{IEEEtran}
{\footnotesize
\bibliography{main}}

\end{document}